\documentclass[twoside,11pt]{article}

\usepackage[preprint]{jmlr2e}

\usepackage{newtxtext}
\usepackage{amsmath}
\usepackage{newtxmath}
\usepackage{bm}
\usepackage{comment}
\usepackage{upgreek}
\usepackage{lastpage}
\newtheorem{assumption}[theorem]{Assumption}

\title{Physics-Transfer Learning: A Physics-Constrained Theory for Zero-Shot Generalization Along the Ladder of Complexity}

\author{\name Yingjie Zhao \email zhaoyingjie@mail.tsinghua.edu.cn \\
       \name Zhiping Xu$^\ast$ \email xuzp@tsinghua.edu.cn \\
       \addr Applied Mechanics Laboratory, Department of Engineering Mechanics\\
       Tsinghua University\\
       Beijing 100084, China\\
       $^\ast$Corresponding author}

\jmlrheading{1}{2026}{1--\pageref{LastPage}}{}{}{--}{Yingjie Zhao and Zhiping Xu}
\ShortHeadings{Physics-Transfer Learning}{Zhao and Xu}
\firstpageno{1}

\newlength{\tempdima}
\newcommand{\rowname}[1]
{\rotatebox{90}{\makebox[\tempdima][c]{\textbf{#1}}}}

\newcommand\zfig[1]{{\color{violet}#1}}
\begin{document}

\maketitle

\begin{abstract}

Reliable out-of-distribution generalization remains a central challenge in machine learning, particularly when models are deployed in data-scarce high-fidelity domains.
Existing approaches often rely on statistical distribution alignment, or invariance assumptions that can fail when domain shifts reflect changes in underlying mechanisms.
Inspired by a longstanding paradigm in physics that understands complex systems by starting from tractable idealizations and progressively incorporating additional mechanisms and levels of fidelity, we propose physics-transfer (PT) learning, a physics-constrained learning theory for zero-shot generalization along a ladder of complexity.
PT constructs chains of anchor domains (CADs) that connect low-cost source domains to complex targets and introduces computable feasibility and reliability indicators for assessing transfer before target-domain labels are observed.
Under physical-consistency and stable-representation assumptions, we derive an operational upper bound on target risk in terms of source empirical risk, the PT indicators, and finite-sample errors, and establish convergence as adjacent complexity gaps shrink and data coverage improves.
We evaluate PT on excited-state energy prediction in quantum systems and dynamical-transition prediction in classical systems using neural networks trained without target-domain labels.
Across substantial distribution shifts, PT yields accurate high-fidelity predictions, while the indicators identify incorrect physical mappings before target-domain prediction fails.
These results provide a model-agnostic framework for designing, diagnosing, and controlling generalization in data-scarce learning problems.

\end{abstract}


\section{Introduction}

Energy calculations and dynamics prediction lie at the core of physical mechanics~\citep{tsien1953physical,landau1932theory,zener1932non,fermi1955studies}, particularly in non-equilibrium scenarios.
Energy landscapes govern the stability, reactivity, and functionality of matter, while dynamical processes determine how physical systems evolve, transform, and respond to external stimuli~\citep{han2025NC,onuchic1997theory,hanggi1990reaction,yang2025multicell,zhao2024CRPS}.
Representative examples include excited-state energy calculations in quantum molecular systems and dynamical strain localization in classical dynamical systems, both of which involve property calculations under strongly non-equilibrium conditions.
Accurate prediction of energies and dynamics directly impacts critical applications such as optoelectronic material design, catalysis, energy storage, and structural materials engineering.

However, modern physical systems of practical relevance are characterized by extreme complexity in the seek for high performance and system integrity.  
In advanced domains such as reusable aerospace systems, controlled nuclear fusion, and integrated semiconductor technologies, physical processes are inherently non-equilibrium, multi-scale, multi-physical, and multi-mechanistic in nature~\citep{baiocco2021overview,yin2024two}. 
This complexity leads to a fundamental accuracy-efficiency dilemma.  
High-fidelity analytical models are typically unavailable.  
Numerical simulations, while accurate, are computationally prohibitive for large-scale systems or long-time evolution, especially when quantum-level precision or fine spatiotemporal resolution is required.
Experimental measurements, on the other hand, suffer from limited coverage of the parameter space, high cost, and restricted accessibility in extreme environments, making systematic data acquisition difficult.  
As a result, reliable prediction of complex physical systems remains fundamentally constrained by the trade-off between physical fidelity, computational efficiency, and data availability.
In physical and materials science, this challenge runs through the end-to-end scientific workflow, including the accurate description of potential energy surfaces (PES)~\citep{zhang2018PRL}, the representation and characterization of material structures~\citep{zhao2023JMPS}, the prediction of non-equilibrium material properties~\citep{zhao2025matter}, and ultimately the spatiotemporal evolution of non-equilibrium dynamics~\citep{zhao2025iclr}.
A unified theory framework therefore needs to be established to encompass these challenges as different manifestations of complexity in physical systems.

Recent advances in artificial intelligence (AI) have introduced new paradigms for addressing these challenges.  
The emergence of large-scale datasets, together with the rapid development of deep learning architectures and optimization algorithms, has enabled data-driven modeling approaches that achieve remarkable performance across scientific domains~\citep{zhang2018PRL,li2022NCS,wang2023scientific,ghafarollahi2025automating}.
In particular, physics-informed machine learning (PIML) has demonstrated success by embedding physical structures into neural networks (NNs) through partial differential equations (PDEs), conservation laws, symmetry constraints, and invariance principles, leading to improved generalization and interpretability~\citep{karniadakis2021NRP,raissi2019JCP,batzner2022NC}.
Despite these successes, fundamental limitations remain. Many real-world physical systems do not admit well-defined governing equations, or are governed by mechanisms that cannot be compactly formulated as explicit PDEs that are needed by these conventional approaches with strong physics (in mathematical forms) constraints~\citep{mcgreivy2024NMI}.
In highly complex, multi-mechanistic, and strongly coupled systems, physical laws may exist only implicitly, across scales, or in hierarchical forms that are not directly accessible to traditional physics-informed learning formulations.
This limits the applicability of equation-based constraints in complex scenarios.
Moreover, high-fidelity digital libraries (DLs) are still far from sufficient.
Existing physics datasets that integrate simulation and experimental data remain scarce and are typically limited to the terabyte scale~\citep{ohana2024well,chow2025physbench,puyin2026quantiphy}.
Recent progress in physical foundation models further highlights both the promise and the limitations of large-scale data-driven learning in physical systems~\citep{chu2026physics}.
Although such models can learn reusable representations across certain physical regimes, their generalization remains strongly dependent on the coverage of training data, the underlying physical regime, temporal scale, and the complexity of target problems~\citep{chu2026physics}.
Therefore, reliable prediction in complex physical systems requires more than data scaling or model scaling, calling instead for an explicit mechanism to organize and evaluate physical knowledge across domains.
As a result, what we can do is to start from simplified models and proceed along the ladder of complexity by transferring the physics and assessing its generalizability and convergence, as has long been practiced in physics and mechanics through successive refinements from idealized models to realistic systems.

We propose the concept and theory of \emph{physics-transfer (PT) Learning}, which introduces physical constraints and data-encoded inductive biases directly into learning theory (\zfig{Figure 1a})~\citep{zhao2025matter,zhao2025iclr,zhao2026predicting}.
By reformulating learning through physically grounded mappings and hierarchical physical structures, this framework enables systematic knowledge transfer across domains of different fidelity and model complexity, thereby supporting zero-shot predictions in unseen target scenarios where direct high-fidelity data are limited or unavailable.
Successful applications have been demonstrated in complex scientific scenarios such as high-performance alloy design and fabrication, as well as predictive modeling of brain morphological development, illustrating its potential for addressing real-world scientific complexity~\citep{zhao2025matter,zhao2025iclr,zhao2026predicting}.
These successes highlight the necessity of a theoretical foundation of the PT approach, as well as demonstration and discussion on simpler, well-controlled physical systems.

In this work, we establish a unified PT theory framework for reliable prediction and generalization in complex physical systems.  
We construct chains of anchor domains (CADs) along the complexity ladder, ordered by physical fidelity and complexity, that bridge between low-cost source domains and high-accuracy target domains.
Importantly, the construction of PT and CADs is not a technical heuristic for improving cross-domain prediction.
Rather, PT and the associated CADs explicitly encode into learning theory a long-standing research paradigm of physics, mechanics, and materials science, in which complex systems are understood by starting from tractable idealizations and progressively incorporating additional physical mechanisms, degrees of freedom, and levels of fidelity.
In contrast to physical foundation models, which often depend on large-scale datasets to capture implicit and sometimes blurred statistical patterns, PT builds on explicitly constructed CADs that encode the underlying physics across complexity levels (\zfig{Figure 1a}).
Based on the CADs, the physically motivated indicators can be calculated to quantitatively evaluate both the feasibility and the convergence of PT, enabling systematic diagnosis of generalization reliability and early detection of failure risks.
Moreover, the target-domain prediction error under out-of-distribution conditions is explicitly bounded by the PT indicators and finite-sample errors.
We validate the framework through two representative classes of physical problems.
These two problems are selected because transition nature combines severe data scarcity with the prohibitive cost of high-fidelity modeling.
The quantum calculations for excited-state energies are exemplified by oscillator strength prediction in molecular systems, and dynamical strain localization in classical dynamical systems is exemplified by sliding resistance calculations.
These two cases represent fundamental physical mappings in energy calculations and dynamics prediction and serve as controlled testbeds for distribution shift, label shift, and zero-shot extrapolation.
Our results demonstrate that, under weak physical constraints, NNs can learn stable latent physical representations that enable accurate prediction across multiple fidelity levels, achieving high-accuracy target-domain prediction using only low-fidelity training data.  
To systematically evaluate when such physics transfer is feasible and reliable, we introduce PT indicators that provide quantitative guidance for generalizability assessment, convergence analysis, and failure diagnosis.
In this broader sense, the unified PT framework encompasses the complexity of physical and materials science by formulating PES, material-structure representations, non-equilibrium properties, and spatiotemporal dynamical evolution within the same theoretical structure.
Together, this work establishes PT learning as a principled learning paradigm that integrates physical structures and the learning theory, offering a scalable, interpretable, and theoretically grounded framework for low-cost, high-accuracy prediction in complex physical systems.

\section{The Concept of Physics Transfer}

Traditional machine learning (ML) paradigms typically rely on the assumptions of independent and identically distributed data and the availability of sufficiently labeled datasets~\citep{abu2012learning}.
However, in complex problems of physical sciences, there often exists distribution shifts, and severe scarcity of labeled data in the target domain, rendering purely statistical learning (SL) frameworks fundamentally limited in achieving reliable generalization.
From a learning-theoretic perspective, domain adaptation theory provides a formal framework for analyzing generalization under distribution shift~\citep{ben2010ML,long2015ICML,wu2025JMLR}.
While this bound formalizes the difficulty of learning under distribution shift, purely statistical alignment of distributions is fundamentally insufficient for complex physical systems, where changes in distributions are driven by evolving physical mechanisms rather than sampling noise.

To address this challenge, this work presents a PT framework~\citep{zhao2025matter,zhao2025iclr}, which introduces physical mappings as weak physical constraints and embeds physical mechanisms directly into the cross-domain learning process along the ladder of complexity (see Appendix~\ref{app:methods} for details).
Instead of relying solely on statistical distribution alignment~\citep{long2015ICML,wu2025JMLR}, PT establishes generalization through physical structural consistency across domains.
The central idea is to reformulate conditional relations not as purely statistical mappings, but as physically grounded mappings $\mathbf{y}_k^i = \mathcal{P}\!\left( \mathbf{x}_k^i; \xi_{\mathcal{A}_k} \right)$, $ \mathbf{x}_k^i$, $\mathbf{y}_k^i$ denotes the input-output pair constrained by $\mathcal{P}$, and $\xi_{\mathcal{A}_k}$ characterizes the fidelity level of the physical model.
To bridge the gap between the source and target domains, PT constructs a sequence of intermediate domains with progressively increasing physical complexity, referred to as CADs,
$\widehat{\mathcal{A}}_k = \left\{ \left( \mathbf{x}_k^i, \mathbf{y}_k^i \right) \right\}_{i=1}^{l_k}, ~ \xi_{\mathcal{A}_k} \in [\xi_{\mathcal{S}}, \xi_{\mathcal{T}}]$, where $k=1,\dots,K$ indexes the anchor domains ordered by increasing physical complexity, and $\xi_{\mathcal{S}}$ and $\xi_{\mathcal{T}}$ denote the fidelities of the source ($\mathcal{S}$) and target ($\mathcal{T}$) domains, respectively.
These anchor domains provide physically interpretable bridges between the source domain and the target domain, enabling a smooth transition from low-fidelity, low-cost models to high-fidelity, high-accuracy physical models while preserving physical coherence (\zfig{Figure 1b}).
Within this framework, physical representations ($Z$) are learned in a latent space of NNs under weak physical constraints enforced by physical mappings, allowing data from different fidelities to naturally align without relying on adversarial alignment~\citep{long2015ICML} or explicit distribution-matching losses~\citep{zhao2019ICLR}.
This avoids the distortion of physical structures that commonly arises in purely statistical domain adaptation approaches.

Meanwhile, physically motivated indicators provide a quantitative mechanism for evaluating feasibility, stability, and risk of cross-domain transfer, enabling predictive diagnosis of model failure rather than post hoc empirical assessment (\zfig{Figure 1c}).
Two indicators are proposed to govern this evaluation process (see Appendix~\ref{app:methods} for details).  
The first indicator, termed as the PT feasibility indicator, $\lambda^{\mathcal{A}_k} = \frac{1}{2} \left[ d(\widehat{{\mathcal{D}}}_{\mathcal{S}}^Y, \widehat{{\mathcal{D}}}_{\mathcal{A}_k}^Y) - d(\widehat{{\mathcal{D}}}_{\mathcal{S}}^Z, \widehat{{\mathcal{D}}}_{\mathcal{A}_k}^Z) \right]^2$, characterizes the feasibility of transferring physical knowledge from the source domain to an anchor domain, thereby quantifying whether meaningful PT can occur across domains
Anchor domains are used as validation intermediates to assess this feasibility.  
The second indicator, termed as the PT reliability indicator, $\delta^{\mathcal{A}_k} = d(\widehat{\mathcal{D}}_{\mathcal{S}}^Z, \widehat{\mathcal{D}}_{\mathcal{T}}^Z)/d(\widehat{\mathcal{D}}_{\mathcal{S}}^Z, \widehat{\mathcal{D}}_{\mathcal{A}_k}^Z)$, characterizes the physical discrepancy between an anchor domain and the target domain, thereby quantifying whether conclusions drawn from anchor-domain validation are informative and reliable for target-domain prediction.
Taken together, $\lambda^{\mathcal{A}_k}$ evaluates whether physics transfer from the source domain to an anchor domain is feasible, while $\delta^{\mathcal{A}_k}$ assesses whether the anchor-domain validation remains predictive for the target domain.
The two indicators with $\lambda^{\mathcal A_k}\geq0$ and $\delta^{\mathcal A_k}\geq1$, provide complementary diagnostics for guiding the PT learning process.

\begin{figure}[t]
\centering
\includegraphics[width=\linewidth]{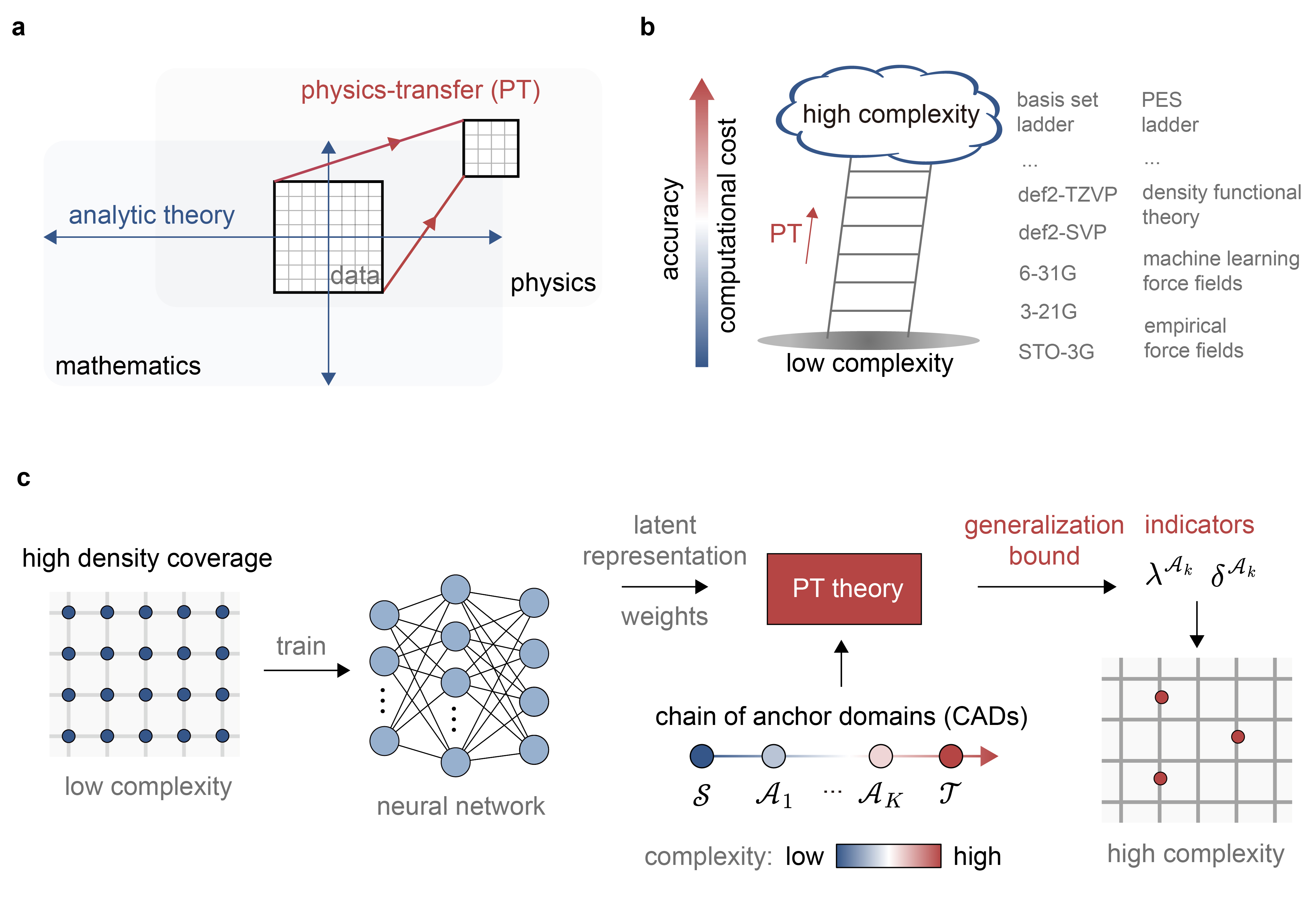}
\caption{\textbf{Physics-transfer (PT) theory framework.}
\textbf{(a)} The concept of PT.
\textbf{(b)} A broadly defined complexity ladder, encompassing aspects such as exchange-correlation functionals, basis sets, and potential energy surfaces (PES).
\textbf{(c)} Schematic of the PT learning framework. Comprehensive digital libraries (DLs) with high-density, large-coverage data enable machine learning (ML) models to learn fundamental physics in simplified or well-characterized settings; the learned physics is then transferred and adapted to more complex, data-scarce scenarios via chains of anchor domains (CADs) and PT indicators.}
\label{fig1}
\end{figure}

In this work, the feasibility evaluation follows a structured protocol (see Appendix~\ref{app:methods} for details).
A NN model is trained exclusively on the source domain, after which the network weights are fixed, and the model is directly evaluated along the CAD $\mathcal{A}_{k}$.
Small values of $\lambda$ and $\delta$ indicate high feasibility of PT.
The convergence evaluation follows a complementary protocol.
A NN is independently trained on anchor domains of progressively increasing complexity.
Convergence is assessed on the most complex anchor domain by analyzing whether predictive performance stabilizes as domain complexity increases.
Models that exhibit stable convergence behavior are then transferred to the target domain, ensuring both robustness and reliability of prediction.
When these conditions are satisfied, models validated on anchor domains can be reliably applied to the target domain, enabling high-accuracy prediction under low computational cost.


Compared with conventional domain adaptation and domain generalization methods~\citep{ben2010ML,cortes2014TCS,long2015ICML,zhao2019ICLR,jiang2022PR,he2024JMLR,wu2025JMLR}, PT learning does not depend on target-domain labels, nor does it assume invariant conditional distributions.
Instead, it uses the consistency of underlying physical mechanisms as the fundamental basis for generalization, thereby enabling reliable learning under distribution shift, label shift, and zero-shot extrapolation scenarios that have long posed significant challenges for SL theory~\citep{abu2012learning}.
PT is complementary to transfer learning, which reuses learned representations across tasks or domains, and to geometric deep learning, which embeds geometric structures and symmetries into model architectures.
PT further provides theoretical guarantees for generalization by leveraging physical continuity along the CADs and the completeness of the DLs (see Appendix~\ref{app:methods} for details).
This paradigm provides a scalable and physically interpretable pathway for achieving low-cost, high-accuracy prediction in complex physical systems.

To illustrate the core ideas and applicability of PT learning, we select two representative physical problems as examples, corresponding respectively to excited-state energy calculations and dynamical strain localization.
These examples demonstrate physics transfer across fidelity levels, where excited-state energy calculations are transferred across different basis-set fidelities, while strain localization dynamics are transferred across force-field fidelities.
In practical applications, both types of problems share the common accuracy-efficiency challenge.
High-fidelity physical models are computationally expensive, while low-fidelity models, although efficient, suffer from systematic bias.
In the following sections, we introduce these two examples in turn and demonstrate how PT learning achieves reliable generalization across different types of physical mappings and addresses the accuracy-efficiency dilemma.

\section{Results for Representative Problems}

These two representative problems are selected because transition phenomena commonly involve severe data scarcity and prohibitively expensive high-fidelity modeling.
The first concerns electronic transitions in quantum systems, where accurate excited-state calculations require increasingly costly quantum-chemical methods.
The second concerns dynamical transitions in classical systems, where strain localization and sliding resistance depend sensitively on the fidelity of the underlying PES.
Together, they provide complementary testbeds for evaluating PT across distinct physical mappings and fidelity hierarchies.

\subsection{PT Models for Quantum Calculations of Excited-State Energies}

The first example focuses on the prediction of oscillator strength in molecular systems~\citep{vinod2025qemfi}.  
Oscillator strength is a key physical quantity that characterizes the probability of electronic transitions and directly determines the intensity of molecular absorption spectra.
It plays an important role in fields such as spectral analysis, optoelectronic material design, and molecular recognition.  
From a physical perspective, oscillator strength is not an independent empirical quantity, but is jointly determined by the transition energy and the transition electric dipole moment.
Its fundamental relationship, which follows from time-dependent perturbation theory and Fermi’s Golden Rule under the electric-dipole approximation~\citep{mchale2017molecular}, can be written as
\begin{equation*}
f \propto \Delta E \, \big| \boldsymbol{\upmu} \big|^2 ,
\end{equation*}
where $\Delta E$ is the transition energy and $\boldsymbol{\upmu}$ is the transition electric dipole moment.  

Within the PT learning framework, the physical mapping $\mathcal{P}$ is defined as the physical relationship from the transition energy and transition electric dipole moment to the oscillator strength, that is $f = \mathcal{P}(\Delta E, \boldsymbol{\upmu})$ (\zfig{Figure 2a}).  
The CADs are constructed progressively along the levels of quantum chemical computational fidelity (\zfig{Figure 2b}).
The source domain $\mathcal{S}$ adopts the low-cost Slater-type orbital (STO-3G) basis set.
The anchor domains $\mathcal{A}_1$, $\mathcal{A}_2$, and $\mathcal{A}_3$ correspond to the 3-21G, 6-31G, and def2-SVP methods, respectively.
The target domain $\mathcal{T}$ is defined as the def2-TZVP method~\citep{weigend2005PCCP}.  
As the basis set accuracy increases, the computational cost rises rapidly due to the growth of the basis dimension and the scaling of electronic structure calculations (typically between $\mathcal{O}(N^3)$ and $\mathcal{O}(N^6)$ with respect to the number of basis functions $N$), making high-fidelity calculations substantially more expensive than minimal basis sets such as STO-3G~\citep{vinod2025qemfi,weigend2005PCCP}.
Meanwhile, the distribution of oscillator strength exhibits systematic shifts~\citep{vinod2025qemfi}.  

We first evaluate the PT feasibility and reliability for the physical mapping across different fidelities.
Specifically, we train the ML model only on the source domain to learn the underlying physical relations embedded in the data.  
The parameters obtained under weak physical constraints during training are then fixed, and the well-trained model is directly applied to prediction in the anchor domains.  
Based on the predictive performance on the anchor domains and the PT indicators, the feasibility of PT can be evaluated.  
If weak physical transfer is feasible, prediction on the target domain can be further performed, thereby achieving high-chemical-accuracy oscillator strength prediction under low computational cost conditions.  
To quantitatively evaluate the feasibility and reliability of the PT process, we compute the PT indicators $\lambda^{\mathcal{A}_k}$ and $\delta^{\mathcal{A}_k}$.  
For the PT model trained on the source domain, we calculate the indicators on the three anchor domains with increasing chemical complexity, and report the indicator values as mean $\pm$ standard deviation (\zfig{Figure 2c}).
Specifically, the PT feasibility indicator increases from $\lambda^{\mathcal{A}_1}=(7.0\pm13.0)\times10^{-6}$ to $\lambda^{\mathcal{A}_2}=(7.09\pm1.16)\times10^{-4}$ and $\lambda^{\mathcal{A}_3}=(1.65\pm0.31)\times10^{-3}$, showing a monotonic growth of the source-anchor discrepancy with increasing basis-set fidelity (\zfig{Figure 2d}).
Nevertheless, even for the anchor domain ($\mathcal{A}_3$) with the largest discrepancy from the source domain, $\lambda^{\mathcal{A}_3}$ remains below $2\times10^{-3}$, indicating that the physical mapping learned from the STO-3G source domain remains transferable across the CADs.
In contrast, the PT reliability indicator decreases from $\delta^{\mathcal{A}_1}=15.11\pm4.21$ to $\delta^{\mathcal{A}_2}=4.23\pm0.16$ and $\delta^{\mathcal{A}_3}=2.72\pm0.06$, corresponding to an approximately $82\%$ reduction from $\mathcal{A}_1$ to $\mathcal{A}_3$.
This decrease indicates that high-accuracy anchor domains $\mathcal{A}_3$ become progressively closer to the target domain in the learned physical representation space, so predictive performance on $\mathcal{A}_3$ provides a more reliable proxy for target-domain feasibility than performance on lower-fidelity anchors.
By directly applying the physical relationships learned in the source domain to the anchor domains, the computational results show that even in the presence of significant label distribution shifts between the anchor domains and the source domain, the model trained on the source domain can still accurately predict the oscillator strength across all anchor domains.  
Specifically, even for the anchor domain with the largest discrepancy from the source domain ($\mathcal{A}_3$), the coefficient of determination $R^2$ between model predictions and ground truth values exceeds $0.98$ (\zfig{Figure 2e}).

\begin{figure}[t]
\centering
\includegraphics[width=\linewidth,height=0.68\textheight,keepaspectratio]{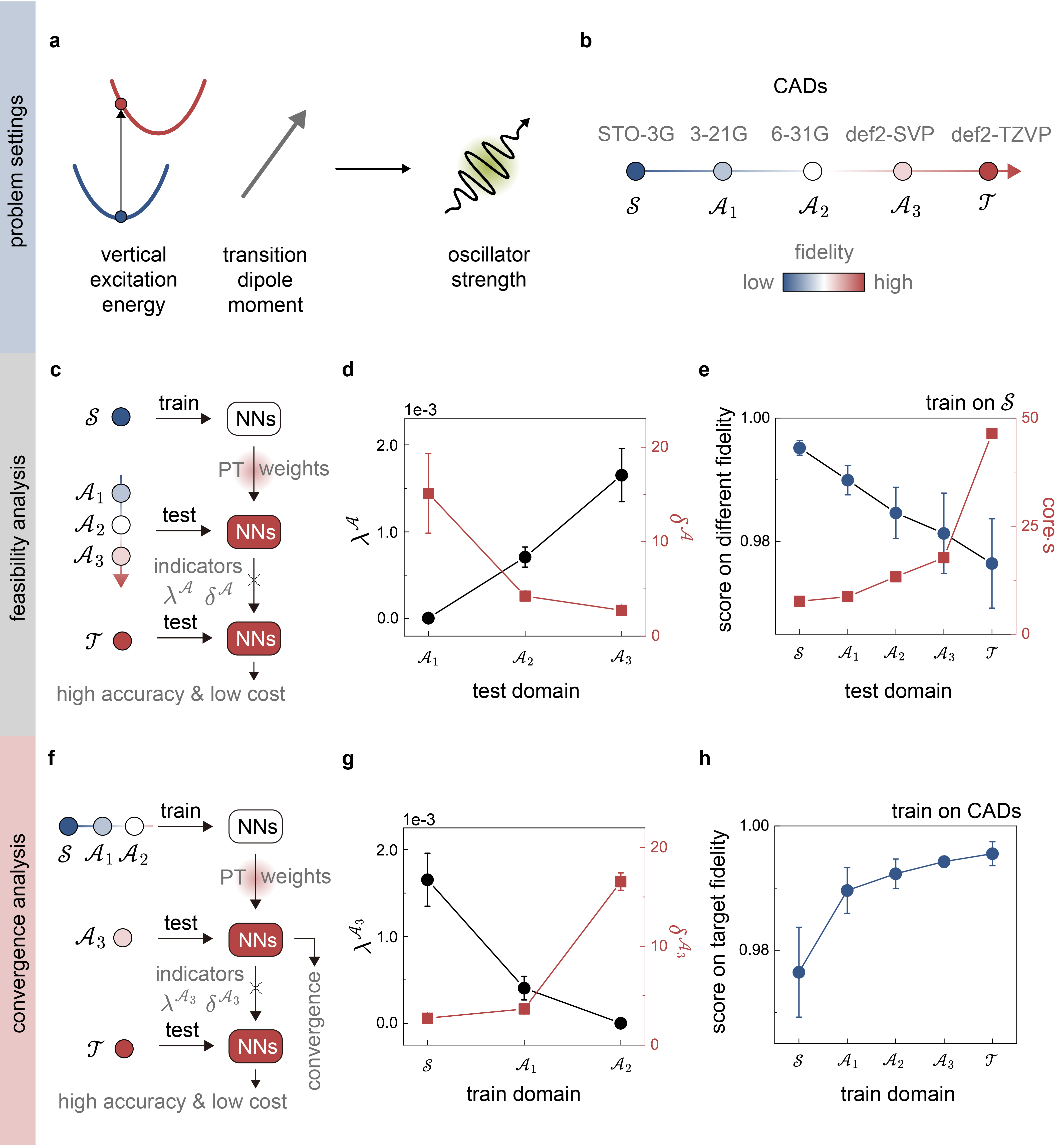}
\caption{\textbf{Application of the PT theory framework to oscillator strength prediction.}
\textbf{(a)} Physical mapping between oscillator strength, transition energy, and dipole moment.
\textbf{(b)} CADs constructed according to the complexity of quantum-chemical basis sets.
\textbf{(c)} Schematic for evaluating the feasibility of PT.
\textbf{(d)} PT indicators for different test domains.
\textbf{(e)} PT performance across different test domains.
\textbf{(f)} Schematic of PT convergence analysis.
\textbf{(g)} Variation of PT indicators with increasing chemical accuracy of the training domain.
\textbf{(h)} PT performance with increasing chemical accuracy of the training domain.}
\label{fig2}
\end{figure}

Furthermore, by gradually increasing the complexity of the training samples, the convergence behavior of PT can be tested.  
Specifically, ML models are trained separately on the source domain and the anchor domains ($\mathcal{A}_1$, $\mathcal{A}_2$).
The trained model parameters are then fixed and directly applied to prediction on the anchor domain ($\mathcal{A}_3$) and the target domain, in order to evaluate the convergence of PT (\zfig{Figure 2f}).  
The results in \zfig{Figure 2g} show a clear quantitative convergence of the feasibility indicator.
When the training domain is increased from $\mathcal{S}$ to $\mathcal{A}_1$ and $\mathcal{A}_2$, $\lambda^{\mathcal{A}_3}$ decreases from $(1.65\pm0.31)\times10^{-3}$ to $(4.06\pm1.36)\times10^{-4}$ and finally to $(1.0\pm1.0)\times10^{-6}$.
This corresponds to a reduction of about $75\%$ from $\mathcal{S}$ to $\mathcal{A}_1$ and more than $99.9\%$ from $\mathcal{S}$ to $\mathcal{A}_2$, confirming that increasing the chemical accuracy of the training data rapidly reduces the physical discrepancy to $\mathcal{A}_3$ and enhances the feasibility of PT.
It should be noted that $\delta^{\mathcal{A}_3}$ gradually increases with the improvement of chemical accuracy in the training data, which quantifies the discrepancy between the anchor domain $\mathcal{A}_3$ and the target domain $\mathcal{T}$.
Specifically, $\delta^{\mathcal{A}_3}$ increases from $2.72\pm0.06$ to $3.67\pm0.11$ and $16.55\pm0.88$, an approximately six-fold increase relative to the source-trained case.
This opposite trend indicates that, although PT to $\mathcal{A}_3$ becomes nearly converged, the relative discrepancy between $\mathcal{A}_3$ and the target domain $\mathcal{T}$ remains an independent reliability factor that must be monitored before using anchor-domain success as evidence for target-domain prediction.
As the chemical accuracy of the training data increases, the coefficient of determination $R^2$ between model predictions and ground truth values also exhibits clear convergence behavior.  
When the model is trained on higher-accuracy anchor domains, its predictive performance on the target domain continues to improve, with $R^2$ exceeding $0.99$, approaching the interpolation-level prediction accuracy within the target domain (\zfig{Figure 2h}).

By integrating the PT indicators, the model’s predictive performance on the anchor domains, and the convergence analysis of PT, sufficient confidence can be established in the reliability of prediction on the target domain.  
The experimental results further validate the theoretical analysis.
On the target domain, even in the presence of significant label distribution shifts, the model still achieves high-accuracy prediction, with the coefficient of determination $R^2$ exceeding $0.96$ (\zfig{Figures 3a and 3b}).  
Meanwhile, in the latent physical representation space, the features of the source domain and the target domain exhibit good alignment, indicating that the model has successfully learned stable underlying physical mechanisms in the data, rather than merely performing statistical fitting (\zfig{Figure 3c}).
The introduction of weak physical constraints enables data of different fidelities to naturally align in the latent space, thereby allowing stable extrapolative prediction even under label shift conditions.  

In contrast, when such weak physical constraints are absent, physical transfer inevitably fails.  
For example, if the ground-state electric dipole moment ($\boldsymbol{\upmu}$) is used instead of the transition electric dipole moment to predict oscillator strength, the lack of a clear physical relationship between the two prevents the model from learning a unified physical representation.
As a result, data of different fidelities cannot be aligned in the latent space (\zfig{Figure 3d}).  
In this case, the PT indicators behave as follows, even if some anchor domains are relatively close to the target domain (small $\delta^{\mathcal{A}_k}$), their physical discrepancy from the source domain becomes significantly larger (large $\lambda^{\mathcal{A}_k}$), and the model’s predictive performance on the anchor domains already shows clear failure (\zfig{Figure 3e}).  
According to the abnormal PT feasibility indicator, failure of prediction on the target domain can therefore be diagnosed in advance, and the computational results are fully consistent with this theoretical judgment (\zfig{Figure 3f}).
The above example of PT failure demonstrates that, even when ground-truth data are unavailable in practice, the risk of PT failure can be identified in advance through the analysis of PT learning indicators and convergence behavior.  
This diagnosis can further be used to guide successful PT by gradually increasing the complexity of the training data or by correcting the physical mapping (\zfig{Figure 3g}).
Furthermore, if experimental data of $f = \mathcal{P}(\Delta E, \boldsymbol{\upmu})$ is introduced as the ground truth, one can then quantitatively assess the methods along the CADs.

\begin{figure}[t]
\centering
\includegraphics[width=\linewidth,height=0.65\textheight,keepaspectratio]{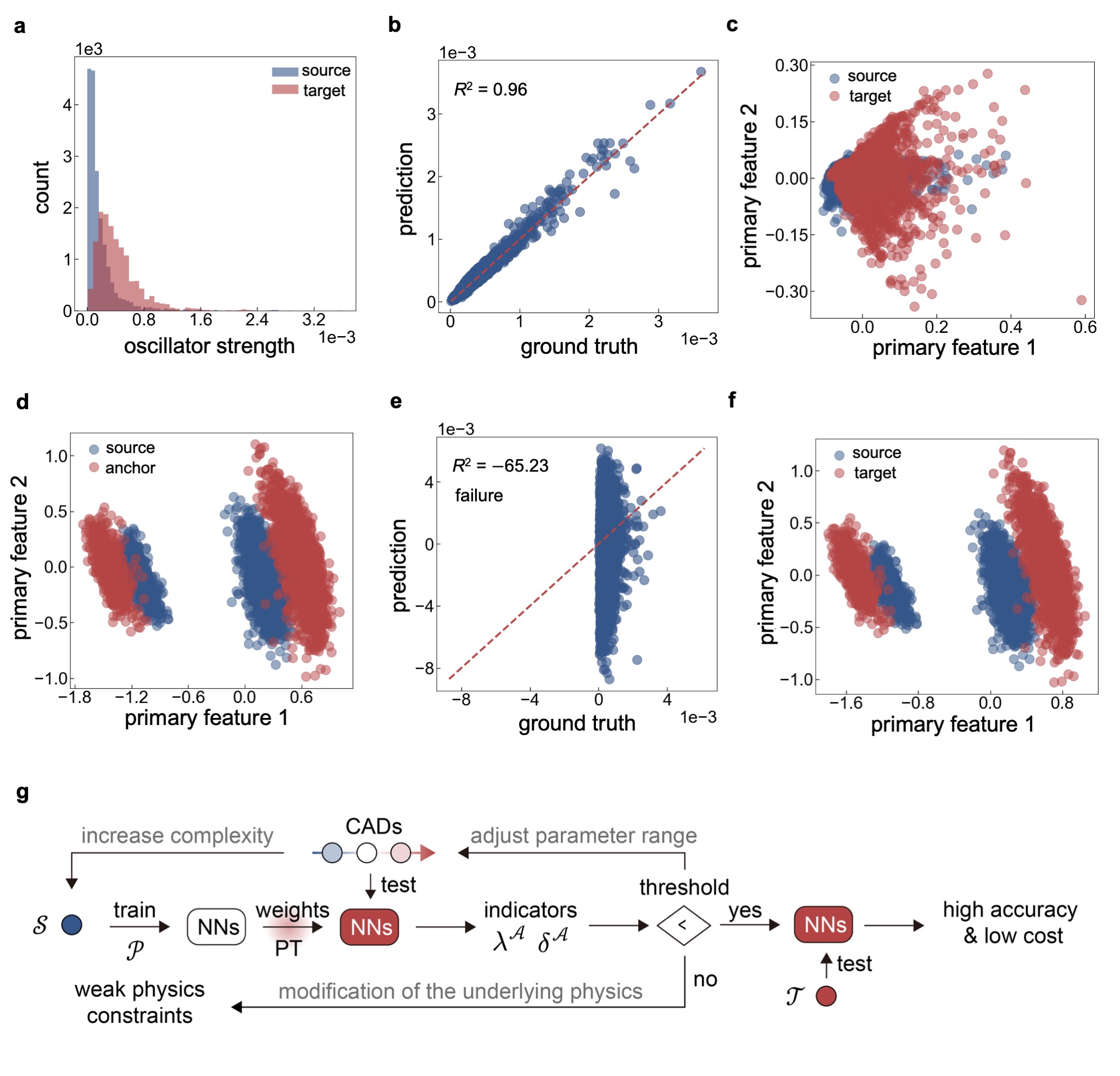}
\caption{\textbf{PT prediction performance and generalization enhancement strategies.}
\textbf{(a--c)} Accurate and efficient prediction of oscillator strength.
\textbf{(a)} Distribution shift of oscillator strength between the source and target domains.
\textbf{(b)} Accurate PT prediction of oscillator strength in the target domain.
\textbf{(c)} Distribution of primary features of the source and target domains in the latent space.
\textbf{(d--f)} Failure of PT caused by an improper physical-mapping constraint.
\textbf{(d)} Distribution of primary features of the source domain and anchor domain $\mathcal{A}_3$ in the latent space.
\textbf{(e)} Prediction failure of PT.
\textbf{(f)} Distribution of primary features of the source and target domains in the latent space.
\textbf{(g)} Monitoring the risk of PT failure and guiding successful PT.}
\label{fig3}
\end{figure}

\subsection{PT Models for Dynamical Transition in Classical Systems}

While accurate energy prediction provides fundamental insights into thermodynamic stability, the prediction of dynamical processes is often even more critical, as it governs the transition of physical systems, such as transport, deformation, and failure behaviors of materials.
Therefore, the second example focuses on dynamical strain localization in classical dynamical systems via sliding resistance calculations.
Sliding resistance quantifies the intrinsic difficulty of displacement between crystalline planes.
It dictates the dynamics of dislocation plasticity, encompassing the transition from collective atomic motion to localized strain and soliton-like transport.
Consequently, this quantity serves as a critical descriptor for predicting friction, plastic deformation, and failure of materials.
The Frenkel-Kontorova (FK) model, as a classical model for describing the motion of atomic chains in periodic potential fields, reveals the essential physical relationship between sliding resistance ($F$) and lattice constant ($a$), spring stiffness ($K$), potential barrier ($V_0$) and periodicity ($b$) (\zfig{Figure 3g})~\citep{braun2004frenkel}:

\begin{equation*}
H = \sum_i \left[ \frac{K}{2} (u_{i+1} - u_i - a)^2 + V_0 \cos\left(\frac{2\pi u_i}{b}\right) - F u_i \right].
\end{equation*}

This explicit dynamical physical relationship provides a theoretical foundation for constructing weak physical constraints.  
Within the PT learning framework, the physical mapping $\mathcal{P}$ is defined as the physical relationship from material parameters to sliding resistance.  
Specifically, the model input variables include the lattice constant, spring stiffness, periodic potential barrier height and period, while the output variable is the sliding resistance (\zfig{Figure 4a}).

The CADs are constructed according to the computational fidelity of the PES.
The source domain $\mathcal{S}$ adopts empirical force fields such as embedded atom method (EAM)~\citep{becker2013EAMdata} and modified embedded atom method (MEAM), the anchor domain consists of machine-learning force fields (MLFFs)~\citep{song2024general}, and the target domain corresponds to density functional theory (DFT) calculations~\citep{su2019JAP} (\zfig{Figure 4b}).

According to differences in the distribution of various materials in the parameter space, this prediction task simultaneously involves interpolation, combinational generalization, and extrapolation when transferring from the source domain to the anchor and target domains (\zfig{Figures 4c-4e}).
For several face-centered-cubic (fcc) metals, including Ag, Au, Cu, Ni, Pd, and Pt, the material parameters in the target domain fall within or moderate shifts to the parameter ranges covered by the source domain training data.
For these systems, the prediction task primarily corresponds to interpolation or mild combinational generalization.
To further examine the extrapolative capability of PT, two additional fcc metals, Ir and Rh, are intentionally selected as challenging test cases.
Compared with the metals included in the training distribution, Ir and Rh exhibit substantially higher cohesive energies and stronger bonding strengths, which lead to larger spring stiffness and modified effective potential barriers in the FK models (\zfig{Figures 4c-4e}).
These physical differences result in a noticeable shift of their parameter distributions relative to the source-domain training range.
Consequently, prediction on Ir and Rh corresponds to a more demanding extrapolation problem in the material parameter space.

Although the PT indicators provide a quantitative criterion for assessing transfer feasibility and reliability, their rigorous estimation requires sufficiently large sample sizes in each domain in order to accurately characterize the distribution discrepancies. 
In the dynamical strain localization problem, the number of available force fields is relatively limited, resulting in a small number of samples and making the statistical estimation of distribution divergences less reliable. 
This limitation arises partly from the high computational cost associated with constructing force fields across multiple fidelity levels. 
Furthermore, the CADs structure in the dynamical strain localization problem is sparser than that in the quantum excited-state example. The hierarchy of computational fidelities for PES is typically limited to empirical force fields, MLFFs, and DFT calculations. 
As a result, the number of intermediate anchor domains is relatively small, which further restricts systematic convergence analyses along the CADs. 
Consequently, in this example, the PT indicators are primarily used in a qualitative manner through the visualization of latent representations and the observation of distribution alignment between domains, together with verification from computational results.

\begin{figure}[t]
\centering
\includegraphics[width=\linewidth,height=0.68\textheight,keepaspectratio]{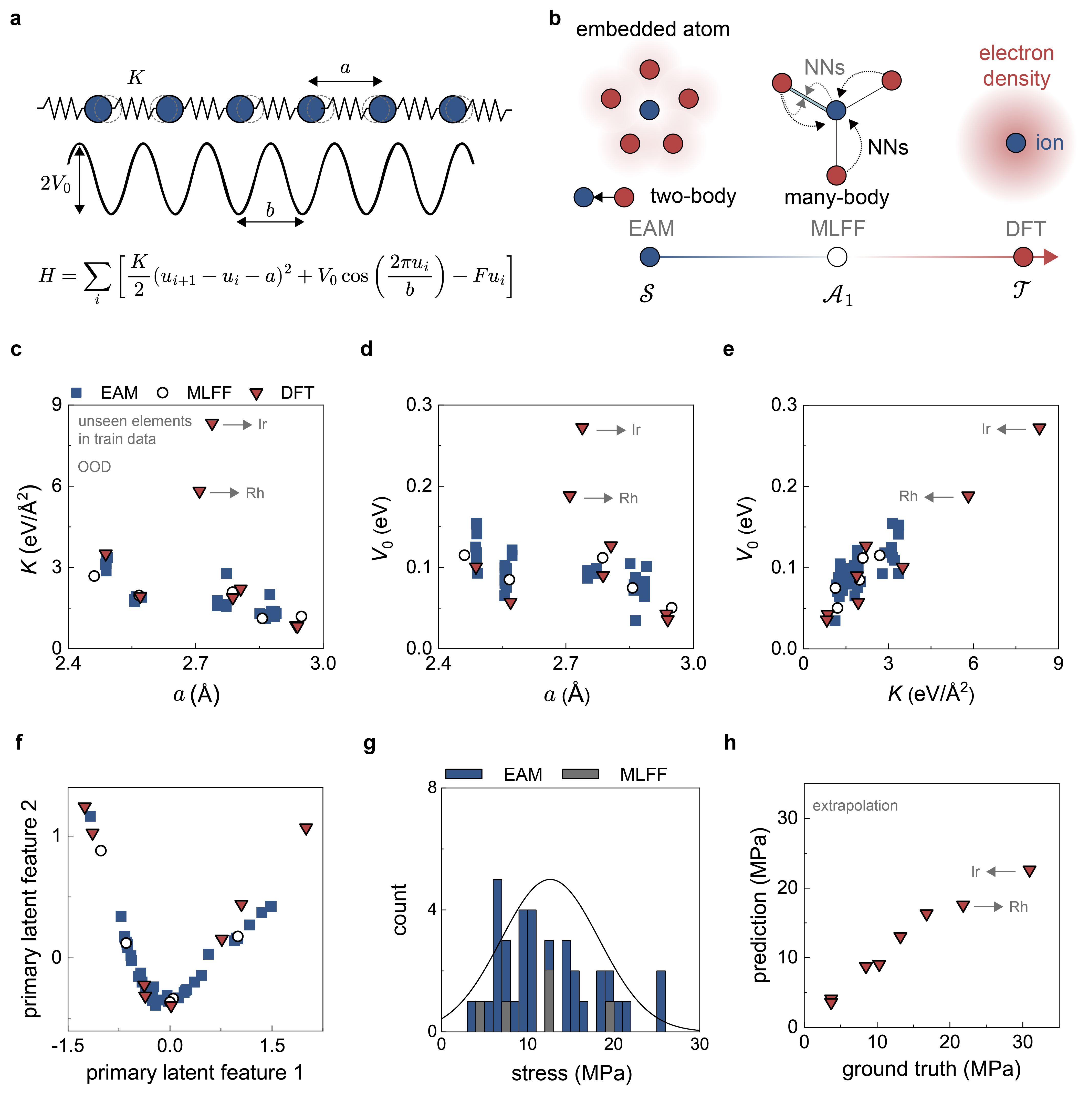}
\caption{\textbf{Application of the PT theory framework to sliding resistance.}
\textbf{(a)} Schematic of the Frenkel--Kontorova (FK) model.
\textbf{(b)} Definition of the CADs for the sliding-resistance problem.
\textbf{(c--e)} Sample distributions in the input parameter space.
\textbf{(c)} Lattice-constant--elastic-constant space.
\textbf{(d)} Lattice-constant--potential-barrier space.
\textbf{(e)} Elastic-constant--potential-barrier space.
\textbf{(f--h)} PT prediction of sliding resistance.
\textbf{(f)} Distribution of primary features in the latent space.
\textbf{(g)} Label distributions of the anchor and source domains.
\textbf{(h)} Prediction results in the target domain.}
\label{fig4}
\end{figure}

Under the guidance of weak physical constraints, the physical representations learned by the model in the latent space enable natural alignment between the source-domain and anchor-domain data (\zfig{Figures 4f and 4g}).
Such alignment in both the latent representation space and the label space provides qualitative evidence consistent with small PT indicators, suggesting that the joint discrepancy between the source and anchor domains remains limited (corresponding to a small $\lambda^{\mathcal{A}_k}$), while the anchor domains also exhibit good physical consistency with the target domain (corresponding to a small $\delta^{\mathcal{A}_k}$) (\zfig{Figures 4f and 4g}).
Quantitatively, the PT feasibility indicator remains small, with $\lambda^{\mathcal{A}_k}=(2.53\pm2.42)\times10^{-5}$, indicating that the discrepancy between the label-space and latent-space source-anchor distances is strongly suppressed by the learned physical representation.
The PT reliability indicator is $\delta^{\mathcal{A}_k}=3.46\pm0.81$, showing that the anchor domain of MLFFs remains within a moderate discrepancy range relative to the target domain of DFT despite the sparse CADs available for the PES hierarchy.
Therefore, sufficient confidence can be established in the model’s predictive performance on both the anchor domains and the target domain.
Consistently, direct prediction on the target domain achieves a mean coefficient of determination of $R^2=0.9898$ over the independent ten runs, demonstrating that the small feasibility indicator and moderate reliability indicator translate into high target-domain accuracy. 
The computational results are highly consistent with the theoretical analysis.
The PT model trained only on source-domain empirical force-field data can accurately predict the sliding resistance at both MLFFs accuracy in the anchor domains and DFT accuracy in the target domain (\zfig{Figure 4h}).

Without such weak physical constraints, successful physical transfer cannot be achieved.
For example, as a contrasting SL baseline, we remove the weak physical constraint associated with the non-equilibrium transition and train the model using only equilibrium-state descriptors, such as lattice constants and elastic moduli, to predict the non-equilibrium strain-localization property.
In this setting, the model no longer receives information that directly encodes the potential barrier, sliding pathway, or transition mechanism governing localized deformation.
The PT feasibility indicator  increases from $(2.53\pm2.42)\times10^{-5}$ for PT to $(8.14\pm11.60)\times10^{-3}$ for SL, corresponding to an increase of more than two orders of magnitude and revealing a much larger mismatch of learned representation between anchor and source domains.
Although the mean reliability indicator of SL, $\delta^{\mathcal{A}_k}=3.30\pm2.76$, is comparable to that of PT, $\delta^{\mathcal{A}_k}=3.46\pm0.81$, its substantially larger fluctuation indicates that anchor-domain evidence becomes less stable when the non-equilibrium physical constraint is absent.
The resulting target-domain performance decreases to $R^2=0.8410$, substantially below the PT result of $R^2=0.9898$.
Thus, for dynamical strain localization, the failure mode is not simply that the target domain is far from the anchor domain, but that equilibrium descriptors alone cannot organize the latent space according to the non-equilibrium transition physics.
This comparison shows that the PT indicators can distinguish physically meaningful transfer from apparent statistical correlation, providing an early warning for unreliable target-domain prediction and guiding the correction of the physical mapping used in model construction.
This result further verifies the effectiveness and universality of PT learning in dynamical problems and suggests its broader applicability to non-equilibrium processes, including strength prediction, as well as transition phenomena involving protein-folding pathways and the associated free-energy barriers rather than equilibrium configurations alone.

\section{Discussion}

\subsection{Challenges in CAD Construction}

A practical limitation of the PT framework lies in the availability and resolution of CADs.
The historical development of engineering science has largely relied on perturbative or gradual improvements in methodological fidelity.
For example, MLFFs perform well for equilibrium properties but remain challenging to apply to non-equilibrium fracture processes, often requiring an order-of-magnitude higher computational cost~\citep{shi2024EML}, while their fidelity strongly depends on the quality and coverage of the training datasets~\citep{wood2025UMA}.
Recent generations of MLFFs have progressively incorporated long-range electrostatics, charge transfer, bond breaking, and chemical reactivity, thereby forming a hierarchy of models with increasing physical fidelity~\citep{shi2024EML}.
When the density of anchor domains is insufficient or the complexity gap between adjacent domains is too large, it becomes difficult to fully characterize the feasibility and convergence of PT using the proposed indicators.  
In such cases, the anchor domains may fail to provide a sufficiently smooth physical bridge between the source and target domains, limiting the diagnostic and predictive power of the PT theory.
This limitation is not intrinsic to the theory itself, but rather reflects the current constraints of data generation methodologies in computational and experimental physics.  
Ongoing advances in computational science and experimental techniques are expected to progressively alleviate this bottleneck.  
For example, in the dynamical strain localization problem, the rapid development of MLFFs with multiple fidelity levels provides a natural pathway for constructing finer-grained CADs, enabling more continuous physical transitions across fidelities~\citep{song2024general}.  
Similarly, advances in experimental characterization techniques make it increasingly feasible to validate PT-based predictions through direct experimental measurements~\citep{zhao2025matter}, thereby establishing a closed-loop verification framework that integrates theory, computation, and experiment.

\subsection{Data, Model, and Algorithmic Perspectives on Enhancing PT}

Effective PT fundamentally depends on whether the source-domain data faithfully encode the underlying physical mechanisms.
The key is to extract the controlling physics from the source domain and transfer it to the target domain.
Data that accurately reflect intrinsic physical relationships provide structurally consistent supervision signals, enabling NNs to learn physically meaningful latent representations rather than spurious statistical correlations~\citep{karniadakis2021NRP}.  
Thus, the quality, physical interpretability, and mechanistic fidelity of training data are critical enablers of successful PT.
In addition, the PT framework imposes no architectural constraints on the choice of ML models.
The performance of PT models can be assessed against ground truth obtained from first-principles calculations or experiments.
It is model-agnostic by design, allowing the integration of diverse architectures ranging from classical ML models to advanced deep learning models~\citep{lecun2015Nature}. 
As ML models continue to evolve in representational capacity and expressivity, their ability to capture latent physical structures improves correspondingly.  
Consequently, PT-based learning naturally benefits from algorithmic advances in representation learning, enabling deeper and more faithful extraction of physical mechanisms from data.
Meanwhile, the PT framework is likewise unconstrained by specific training strategies.  
Advanced optimization techniques, such as active learning, adaptive sampling, uncertainty-aware training, and physics-guided optimization, can be seamlessly integrated to improve sampling efficiency, reduce uncertainty, and enhance model robustness~\citep{pickering2022discovering}.  
These strategies contribute to more stable latent physical representations and more reliable PT across domains.
Future work will focus on developing dimension-reduction theories for latent physical representations, enabling the reconstruction of the underlying physical laws transferred across domains.

\subsection{PT as a Physics-Constrained Learning Theory for Generalization Control}

More fundamentally, PT learning can be interpreted as a physically constrained learning theory for generalization control.  
Instead of treating generalization as a purely statistical phenomenon governed by distribution alignment~\citep{long2015ICML,zhao2019ICLR}, PT explicitly introduces the physical structure across CADs.  
This enables the assessment of generalization risk before deployment through physically motivated indicators, rather than relying on post hoc empirical validation.
By quantifying feasibility, discrepancy, and convergence through physically interpretable indicators, PT provides not only predictive capability but also diagnostic and guiding functionality.  
In this sense, PT does not merely improve prediction accuracy, but establishes a principled framework for risk-aware generalization, reliability assessment, and physically grounded model transfer in complexity.
Overall, PT learning theory provides a unifying paradigm that integrates data fidelity, model expressivity, algorithmic optimization, and physical structure into a coherent learning framework.  
As data generation technologies, modeling techniques, and optimization strategies continue to advance, the effectiveness, resolution, and applicability of PT are expected to expand correspondingly, enabling scalable, reliable, and physically interpretable learning across increasingly complex scientific domains.
In this sense, PT may represent an ultimate solution through stepwise extrapolation and gradual transfer of physical knowledge, following the historical trajectory of engineering sciences.

\section{Conclusion}

In this work, we developed the physics-transfer (PT) theory framework for reliable prediction and generalization in complex physical systems. 
By reformulating conditional learning relations as physically grounded mappings and organizing domains through hierarchical chains of anchor domains (CADs) defined on the complexity ladder, we established a structured pathway for transferring knowledge from low-fidelity source domains to high-fidelity target domains.
Unlike conventional domain adaptation methods that rely primarily on statistical distribution alignment, PT leverages the consistency of underlying physical mechanisms as the fundamental basis for generalization. 
The introduction of physically motivated feasibility and reliability indicators enables quantitative assessment of generalizability and convergence, allowing predictive diagnosis before deployment.
Through representative problems of energy calculations and dynamics prediction, we demonstrated that neural networks trained under weak physical constraints can learn stable latent physical representations that support accurate cross-fidelity prediction even in the presence of distribution and label shifts.
By integrating data fidelity, hierarchical modeling, representation learning, and risk-aware generalization analysis, PT establishes a scalable paradigm for low-cost, high-accuracy scientific prediction.
In this sense, PT learning not only addresses the accuracy-efficiency dilemma in current scientific modeling but also contributes a generalizable theoretical structure for physics-guided artificial intelligence in complex domains.



\acks{This work was supported by the National Natural Science Foundation of China through grants 12425201 and 124B2035, and by the Fundamental and Interdisciplinary Disciplines Breakthrough Plan of the Ministry of Education of China (grant no. JYB2025XDXM205).
The computation was performed on the Explorer 1000 cluster system of the Tsinghua National Laboratory for Information Science and Technology.}

\appendix
\section{Methods}
\label{app:methods}

\subsection{Formalization of Physics-Transfer Theory}

\textbf{Preliminaries.}
Let $\mathcal{X}$ and $\mathcal{Y}$ denote the input and output spaces.
A hypothesis $h:\mathcal{X}\rightarrow\mathcal{Y}$ is selected from a hypothesis class $\mathcal{H}$.
Let $\ell:\mathcal{Y}\times\mathcal{Y}\rightarrow \mathbb{R}$ be a bounded loss function.
For a probability distribution $\mathcal{D}$ over $\mathcal{X}\times\mathcal{Y}$, the expected risk is
$\varepsilon_{\mathcal{D}}(h)=
\mathbf{E}_{(x,y)\sim\mathcal{D}}[\ell(h(x),y)]$.
We consider two domains, source ($\mathcal{S}$) and target ($\mathcal{T}$).
Training data are drawn from $\mathcal{D}_{\mathcal S}$, while predictions are required on $\mathcal{D}_{\mathcal T}$.
In physical and engineering sciences, many complex problems remain intractable for direct solution due to formidable theoretical, computational, or experimental challenges.
However, in practice, these problems are often approached by starting from simplified models and progressively introducing corrections for nonlinearity, many-body interactions, hidden variables, and other sources of complexity.
Physics-transfer (PT) learning formalizes this stepwise strategy by progressively transferring physically meaningful structures from the source domain to achieve zero-shot generalization in the target domain.

\noindent
\textbf{Physical anchor domains.}
To formalize the notion of physical complexity, we introduce a parameterized physical mapping behind the observable data.

\begin{definition}[Physical hierarchical mapping (PHM)]
Let $\mathcal{P}:\mathcal{X}\times\Xi\rightarrow\mathcal{Y}$
be a physical mapping parameterized by a complexity parameter $\xi\in\Xi$.
A domain $\mathcal{D}_\xi$ is induced by $Y=\mathcal{P}(X;\xi)$.
\end{definition}

The source and target domains correspond to parameters $\xi_{\mathcal S}$ and $\xi_{\mathcal T}$.
In practice, the complexity gap between these domains may be large, making direct transfer unreliable.
PT therefore introduces intermediate domains that gradually bridge this gap.

\begin{definition}[Chains of anchor domains, CADs]
An anchor domain $\mathcal{A}_k$ is defined by a parameter
$\xi_{\mathcal A_k}$ satisfying $\xi_{\mathcal S} < \xi_{\mathcal A_k} < \xi_{\mathcal T}$.
A sequence $\mathcal{S} \rightarrow \mathcal{A}_1 \rightarrow \cdots \rightarrow \mathcal{A}_K \rightarrow \mathcal{T}$
is called the chains of anchor domains (CADs).
\end{definition}

The CADs construction reflects the historical development of engineering sciences, where physical knowledge is often extended stepwise from simpler systems to more complex ones.

\begin{definition}[Physical stable representation (PSR)]
A representation map $\Phi:\mathcal{X}\rightarrow\mathcal{Z}$ induces a predictor $h(x)=g(\Phi(x))$, where $g:\mathcal{Z}\rightarrow\mathcal{Y}$.
We denote the physical representation as $Z=\Phi(X)$.
\end{definition}

In PT learning, the representation $\Phi$ is expected to capture latent physical variables that remain relatively stable across domains through weak physics constraints $\mathcal{P}$.

\noindent
\textbf{Physical assumptions.}
To characterize cross-domain transfer in physical systems, we introduce two mild assumptions.

\begin{assumption}[Physical consistency]
\label{ass:physical-consistency}
The PHM $\mathcal{P}(x;\xi)$ is Lipschitz continuous with respect to $\xi$
\begin{equation*}
\|\mathcal{P}(x;\xi_1)-\mathcal{P}(x;\xi_2)\| \le L|\xi_1-\xi_2|,
\end{equation*}
for all $x$.
\end{assumption}

This assumption indicates that the physical system responses vary smoothly with respect to complexity parameters.
When the physical discrepancy between neighboring domains is small, the CADs provide a favorable pathway for transferring physical knowledge from the source domain to the target domain.
In practice, however, physical discrepancies may arise due to modeling approximations, numerical discretization errors, limited experimental resolution, or the well-known simulation-to-reality (sim2real) gap. 
These discrepancies may weaken the continuity of the physical mapping. 
Continuous advances in computational modeling, numerical solvers, and experimental measurement technologies help improve the density of CADs and reduce the gap. 
The long-term goal of PT learning is therefore aligned with the broader effort of closing the sim2real gap and maintaining approximate physical consistency across domains.

\begin{assumption}[Stable representation]
\label{ass:stable-representation}
There exists a representation $\Phi$ such that for any $\xi_1,\xi_2$
\begin{equation*}
d(\mathcal D_{\xi_1}^Z,\mathcal D_{\xi_2}^Z) \le C|\xi_1-\xi_2|,
\end{equation*}
where $d(\cdot,\cdot)$ denotes a distribution divergence.
\end{assumption}

This assumption indicates that physically meaningful latent variables change smoothly across domains through the weak physics constraints $\mathcal{P}$.
In PT learning, these weak physics constraints are implemented across multiple components of the learning process, including the construction of the digital libraries (DLs), the design of the ML architecture, the optimization objective and algorithms. 
Together, these mechanisms guide the model to learn latent representations that reflect the underlying physical structures while avoiding the need for explicitly specified governing PDEs.
In addition, even if the underlying physical system exhibits continuity with respect to the complexity parameter, successful learning of the latent physical representation still depends on the completeness of the DLs.
In practice, the representation map $\Phi$ is approximated from finite observations. 
If the DLs insufficiently cover the variability of the physical system, the learned representation may fail to capture the underlying physical mechanisms, and the stability condition may not hold.
Therefore, the completeness and diversity of the DLs play a crucial role in enabling the learning of physically meaningful latent representations. 
Constructing sufficiently rich DLs is an important prerequisite for reliable PT learning across domains.

\begin{lemma}[Representation Lipschitz property]
\label{lem:representation-lipschitz}
Let $\Phi:\mathcal{X}\rightarrow\mathcal{Z}$ be a representation map that is Lipschitz continuous with constant $L_\Phi$.
Then the divergence between domains in the representation space satisfies
\begin{equation*}
d(\mathcal D_{\mathcal S}^Z,\mathcal D_{\mathcal T}^Z) \le L_\Phi d(\mathcal D_{\mathcal S}^X,\mathcal D_{\mathcal T}^X).
\end{equation*}
\end{lemma}

\begin{proof}
Let $Z=\Phi(X)$.  
Since $\Phi$ is Lipschitz continuous, perturbations in the input space are bounded after mapping to the representation space.
Therefore, the divergence between the pushforward distributions $\mathcal D^Z$ induced by $\Phi$ is bounded proportionally by the divergence between the original distributions in $\mathcal X$.
\end{proof}

\noindent
\textbf{Classical domain adaptation bound.}
The analysis builds upon the following classical result.

\begin{lemma}[Domain adaptation bound; see \citet{ben2010ML}]
\label{lem:domain-adaptation}
For any $h\in\mathcal{H}$,
\begin{equation*}
\varepsilon_{\mathcal T}(h) \le \varepsilon_{\mathcal S}(h) + d(\mathcal{D}_{\mathcal S}^X,\mathcal{D}_{\mathcal T}^X) + \lambda^*,
\end{equation*}
where $\lambda^* = min_{h\in\mathcal{H}}(\varepsilon_{\mathcal S}(h)+\varepsilon_{\mathcal T}(h)).$
\end{lemma}

The ideal joint error term $\lambda^*$ is generally not observable when target labels are unavailable.  

\noindent
\textbf{PT indicators.}
PT therefore introduces operational indicators that approximate transfer feasibility and reliability.

\begin{definition}[PT feasibility indicator]
\begin{equation*}
\lambda^{\mathcal A_k} = \frac{1}{2} \left[ d(\widehat{\mathcal D}_{\mathcal S}^Y, \widehat{\mathcal D}_{\mathcal A_k}^Y) - d(\widehat{\mathcal D}_{\mathcal S}^Z, \widehat{\mathcal D}_{\mathcal A_k}^Z) \right]^2,
\end{equation*}
\end{definition}
where $\widehat{\mathcal{D}}$ represents the empirical distribution associated with finite observations from the domain.
This indicator measures the transfer feasibility of physical knowledge from the source domain to the anchor domain, and its form is inspired by \citet{zhao2019ICLR}.

\begin{definition}[PT reliability indicator]
\begin{equation*}
\delta^{\mathcal A_k} = \frac{d(\widehat{\mathcal D}_{\mathcal S}^Z,\widehat{\mathcal D}_{\mathcal T}^Z)}{d(\widehat{\mathcal D}_{\mathcal S}^Z, \widehat{\mathcal D}_{\mathcal A_k}^Z)}.
\end{equation*}
\end{definition}

This indicator quantifies the relative physical discrepancy between the anchor and target domains along CADs.
It assesses the reliability of extrapolating the feasibility conclusions derived from the anchor domain to the target domain.

\noindent
\textbf{PT operational generalization bound.}
The following result establishes an operational upper bound on the target-domain risk in terms of the source-domain empirical error, the PT indicators, and finite-sample estimation errors.

\begin{lemma}
\label{lem:triangle-discrepancy}
The discrepancy between the source and target domains satisfies
\begin{equation*}
d(\mathcal D_{\mathcal S}^Z,\mathcal D_{\mathcal T}^Z) \le d(\mathcal D_{\mathcal S}^Z,\mathcal D_{\mathcal A_k}^Z) + d(\mathcal D_{\mathcal A_k}^Z,\mathcal D_{\mathcal T}^Z).
\end{equation*}
\end{lemma}

\begin{proof}
The result follows directly from the triangle inequality of the divergence metric, where $d(\cdot,\cdot)$ denotes a probability distance satisfying the triangle inequality (e.g., maximum mean discrepancy (MMD), Jensen-Shannon distance (the square root of Jensen-Shannon divergence), or Wasserstein distance)~\citep{gretton2012kernel}.
\end{proof}

\begin{theorem}[PT operational generalization bound]
\label{thm:operational-bound}
Suppose Assumptions~\ref{ass:physical-consistency} and~\ref{ass:stable-representation} hold.
Then for any hypothesis $h\in\mathcal{H}$,
\begin{equation*}
\varepsilon_{\mathcal T}(h) \le \widehat{\varepsilon}_{\mathcal{S}}(h) + d(\widehat{{\mathcal{D}}}_{\mathcal{S}}^X, \widehat{{\mathcal{D}}}_{\mathcal{T}}^X) +B_1\,\lambda^{\mathcal A_k} + B_2\,\delta^{\mathcal A_k} + O\!\left(\sqrt{\frac{d\log n_{\mathcal{S}}}{n_{\mathcal{S}}}}+\sqrt{\frac{d\log n_{\mathcal{A}_k}}{n_{\mathcal{A}_k}}}  + \sqrt{\frac{d\log n_{\mathcal{T}}}{n_{\mathcal{T}}}}\right),
\end{equation*}
\end{theorem}
where $\widehat{\varepsilon}_{\mathcal{S}}(h)$ denotes the empirical risk evaluated on the source domain, $n$ represents the sample size of the different domains along CADs, and $d$ characterizes the complexity of the hypothesis class $\mathcal{H}$ (e.g., the Vapnik-Chervonenkis (VC) dimension).

\begin{proof}
Starting from the domain adaptation bound in Lemma~\ref{lem:domain-adaptation}~\citep{ben2010ML}.
To analyze the discrepancy term, we introduce the anchor domain $\mathcal A_k$ along the CADs.
Using the triangle inequality of the divergence metric,
\begin{align*}
d(\mathcal D_{\mathcal{S}}^X,\mathcal D_{\mathcal{T}}^X)
&\le d(\mathcal D_{\mathcal{S}}^X,\mathcal D_{{\mathcal{A}_k}}^X)
 + d(\mathcal D_{{\mathcal{A}_k}}^X,\mathcal D_{\mathcal{T}}^X).
\end{align*}
Applying the representation Lipschitz property (Lemma~\ref{lem:representation-lipschitz}), the discrepancies in the representation space satisfy
\begin{align*}
d(\mathcal D_{\mathcal{S}}^Z,\mathcal D_{{\mathcal{A}_k}}^Z)
&\le L_\Phi d(\mathcal D_{\mathcal{S}}^X,\mathcal D_{{\mathcal{A}_k}}^X),\\
d(\mathcal D_{{\mathcal{A}_k}}^Z,\mathcal D_{\mathcal{T}}^Z)
&\le L_\Phi d(\mathcal D_{{\mathcal{A}_k}}^X,\mathcal D_{\mathcal{T}}^X).
\end{align*}

Combining these inequalities and Lemma~\ref{lem:triangle-discrepancy} yields
\begin{align*}
d(\mathcal D_{\mathcal{S}}^Z,\mathcal D_{\mathcal{T}}^Z)
&\le L_\Phi d(\mathcal D_{\mathcal S}^X,\mathcal D_{\mathcal T}^X)\\
&\le L_\Phi \left[
d(\mathcal D_{\mathcal{S}}^X,\mathcal D_{{\mathcal{A}_k}}^X)
+ d(\mathcal D_{{\mathcal{A}_k}}^X,\mathcal D_{\mathcal{T}}^X)
\right].
\end{align*}
Under Assumptions~\ref{ass:physical-consistency} and~\ref{ass:stable-representation}, the discrepancies between domains vary smoothly with the complexity parameter along the CADs.
Consequently, the relative scaling between the discrepancies $d(\mathcal D_{\mathcal S}^X,\mathcal D_{\mathcal T}^X)$ toward the anchor and target domains can be estimated through the PT reliability indicator $\delta^{\mathcal A_k}$.
Meanwhile, the PT feasibility indicator $\lambda^{\mathcal A_k}$ provides an observable surrogate of $\lambda^{*}$ that reflects the PT feasibility between the source and anchor domains~\citep{zhao2019ICLR}.
Combining these bounds with the empirical risk estimation yields the stated inequality.
\end{proof}

\noindent
\textbf{Stepwise transfer across CADs.}
The following result characterizes how stepwise transfer along the CADs controls the accumulation of extrapolation errors under physical continuity and sufficiently complete DLs.

\begin{theorem}[PT extrapolation convergence]
Suppose Assumptions~\ref{ass:physical-consistency} and~\ref{ass:stable-representation} hold. 
Let $\Delta_k = |\xi_{\mathcal A_k}-\xi_{\mathcal A_{k-1}}|$ denote the complexity gap between neighboring domains along the CADs. 
The DLs provide sufficiently dense coverage for each domain so that the learned representation $\Phi$ approximates the underlying physical latent variables with error $\varepsilon_{\text{data}}$.
Then the target-domain prediction error satisfies

\begin{equation*}
\varepsilon_{\mathcal T}(h) \le \widehat{\varepsilon}_{\mathcal{S}}(h) + C_1 \sum_{k=1}^{K+1} \Delta_k^2 + C_2 \varepsilon_{\text{data}}.
\end{equation*}

Consequently, as the CADs become sufficiently dense 
$(\max_k \Delta_k \rightarrow 0)$ and the DLs become sufficiently complete $(\varepsilon_{\text{data}} \rightarrow 0)$, the target-domain prediction error converges to the source-domain error.
\end{theorem}

\begin{proof}
Under Assumption~\ref{ass:physical-consistency}, the PHM varies smoothly with respect to the complexity parameter $\xi$. 
Therefore, the discrepancy between neighboring domains along the CADs is controlled by the complexity gaps $\Delta_k = |\xi_{\mathcal A_k}-\xi_{\mathcal A_{k-1}}|$. 
When the CADs are sufficiently dense, the cumulative physical discrepancy between the source and target domains is bounded by $\sum_{k=1}^{K+1} \Delta_k^2$, which determines the magnitude of the PT feasibility and reliability indicators.
Under Assumption~\ref{ass:stable-representation}, the learned representation approximates the underlying physical latent variables under weak physics constraints when the DLs provide sufficiently dense coverage of the physical variability. 
Finite data coverage introduces a representation approximation error $\varepsilon_{\text{data}}$, which influences the estimation of the PT indicators.
As a result, the discrepancy between the source and anchor domains in the representation space diminishes, which leads to the feasibility indicator satisfying
$\lambda^{\mathcal A_k}\rightarrow 0$.
Meanwhile, the discrepancy ratios between domains along the CADs approach unity, implying
$\delta^{\mathcal A_k}\rightarrow 1$.
Consequently, the PT indicators appearing in Theorem~\ref{thm:operational-bound} can be bounded by the physical complexity gaps along the CADs and the representation approximation error induced by finite data.
\end{proof}

Overall, the PT theory framework converts domain generalization into a controlled extrapolation problem along the CADs and weak physics constraints.
The indicators $\lambda^{\mathcal A_k}$ and $\delta^{\mathcal A_k}$ provide operational criteria for assessing transfer feasibility and reliability.
The PT theory framework also highlights two key conditions for reliable cross-domain generalization, including physical consistency across the CADs and sufficient completeness of the DLs used to learn the latent physical representation. 
When these two conditions are satisfied, stepwise transfer along the CADs provides a principled strategy for achieving robust extrapolation to complex target domains.

\subsection{Complexity Ladders}

The computational levels used in the PT framework correspond to commonly adopted Gaussian-type orbital (GTO) basis sets in quantum chemistry. 
The minimal Slater-type orbital (STO-3G) basis set provides a highly efficient but low-fidelity approximation of molecular orbitals~\citep{hehre1969STO}. 
The split-valence basis sets 3-21G and 6-31G introduce additional flexibility by allowing multiple basis functions for valence electrons, thereby improving the description of electronic structure~\citep{hehre1972self}. 
The def2-SVP and def2-TZVP basis sets belong to the Karlsruhe def2 family of basis sets and provide systematically improved accuracy through the inclusion of polarization and additional valence functions~\citep{weigend2005PCCP}. 
As the basis size increases, both the number of basis functions and the computational cost of electronic structure calculations grow significantly, leading to a hierarchy of computational fidelities that naturally defines the complexity ladder used in the CADs construction.

In atomistic simulations, the fidelity of the potential energy surface (PES) defines a natural hierarchy of computational complexity.
Empirical force fields such as the embedded atom method (EAM) and the modified embedded atom method (MEAM) provide efficient approximations of interatomic interactions based on parameterized functional forms~\citep{becker2013EAMdata}.
Machine-learning force fields (MLFFs) introduce substantially higher representational flexibility by learning the PES directly from quantum mechanical data, enabling improved accuracy across diverse atomic environments~\citep{song2024general}.
At the highest fidelity level, density functional theory (DFT) calculations explicitly solve the electronic structure problem within an approximate quantum mechanical framework, providing significantly more accurate but computationally expensive descriptions of interatomic interactions~\citep{su2019JAP}.
These levels, therefore, form a natural computational fidelity ladder for PES modeling, spanning empirical models, data-driven surrogates, and first-principles calculations.

More broadly, similar complexity ladders arise across many scientific systems.
For example, geometric complexity may evolve from simple spheres to ellipsoids and further to surfaces with nonzero Gaussian curvature, such as saddle geometries and brain morphology~\citep{zhao2026predicting}.
Such progressively increasing structural complexity illustrates the ubiquity of stepwise physical extrapolation across scientific domains.

\subsection{Neural Network Architectures}
Both representative problems are implemented with fully connected neural networks containing an explicit representation-learning block and a scalar regression head.
This choice keeps the learning architecture simple, so that the transfer behavior can be attributed primarily to the physical descriptors and the CADs construction rather than to a highly specialized model design.
For the oscillator-strength problem, the input consists of the transition energy and the three Cartesian components of the transition electric dipole moment.
These four physically motivated descriptors are first mapped to a latent representation through a 64-unit hidden layer and a 32-dimensional representation layer.
The decoder then maps the latent representation through another 64-unit hidden layer to the predicted oscillator strength.
Rectified linear unit (ReLU) activations are used in the hidden layers.
For the sliding-resistance problem, the input consists of the descriptors used to characterize the lattice and potential-energy surface.
The network uses three 64-unit hidden layers before a linear output layer for the scalar sliding resistance.
In both examples, the latent representation before the final regression layer is used as $Z$ for the PT-indicator analysis.

\subsection{Data Sources and CAD Construction}

For oscillator-strength prediction, the data are taken from the QeMFi excited-state dataset~\citep{vinod2025qemfi}.
The source, anchor, and target domains are defined by increasing basis-set fidelity, including STO-3G, 3-21G, 6-31G, def2-SVP, and def2-TZVP.
Transition energies, transition electric dipole moments, and oscillator strengths are used to form paired samples at each fidelity level.
The source domain is used for model training, the anchor domains are used for PT-indicator analysis, and the def2-TZVP domain is reserved as the high-fidelity target.
For sliding-resistance prediction, the data are constructed from descriptions of sliding resistance across force-field fidelities.
Empirical force fields define the source domain, MLFFs define the anchor domain, and DFT calculations define the target domain.
The PT descriptors encode non-equilibrium transition information associated with the sliding pathway and potential barrier.
A reduced statistical learning (SL) control is also considered by using descriptors that do not explicitly encode the non-equilibrium transition mechanism.

\subsection{Training Protocols}

The two examples follow the same training and application protocol.
Models are trained only on the source domain, and all preprocessing parameters are fitted using source-domain data.
The normalization is then applied unchanged to anchor and target domains.
This prevents target-domain information from entering the training process and keeps the anchor and target evaluations in a zero-shot setting.
For oscillator-strength prediction, each fidelity-level dataset is split into training and validation subsets with a 70/30 ratio.
The source-domain model is trained over ten independent runs with mini-batch optimization and early-stopping monitoring on the validation subset.
The trained source models are then fixed and evaluated directly on the anchor and target domains.
For sliding-resistance prediction, ten independent source-domain training runs are also performed.
In each run, 20\% of the source-domain samples were held out for validation.
The model is optimized using the mean-squared-error loss and the adaptive moment estimation (ADAM) optimizer with hyperparameter settings: a learning rate of $0.01$, and first- and second-moment decay rates $\beta_0 = 0.9, \beta_1 = 0.999$~\citep{kingma2014ADAM}.
The difference in training length between the two examples reflects the different sizes and numerical characteristics of the quantum and mechanical datasets, while the source-only training and fixed-model transfer protocol are kept the same.

\subsection{Prediction Metrics and PT Indicators}

Prediction accuracy is evaluated using MAE, RMSE, and $R^2$.
The reported values are summarized over ten independent runs as the mean and standard deviation.
The PT indicators are calculated from empirical distribution distances in the label space and in the learned latent representation space.
In the implementation, these distances were estimated by maximum mean discrepancy (MMD) using a radial-basis-function (RBF) kernel, with the kernel bandwidth selected by the median heuristic~\citep{gretton2012kernel}.
$\lambda^{\mathcal{A}_k}$ compares how far an anchor domain is from the source domain in the output space with how far it is in the learned representation space.
A small value indicates that the learned latent representation preserves the relevant physical relation across the source-anchor shift.
The reliability indicator $\delta^{\mathcal{A}_k}$ compares the source-target latent discrepancy with the source-anchor latent discrepancy.
Together, these two indicators are used to judge whether anchor-domain success provides meaningful evidence for target-domain generalization.

\subsection{Zero-Shot Extrapolations}

Zero-shot extrapolation is evaluated by training the models only on the source domain and then applying the fixed model directly to anchor and target domains.
In the oscillator-strength example, models trained on STO-3G data are evaluated without retraining on the 3-21G, 6-31G, def2-SVP, and def2-TZVP domains.
For the convergence analysis, models are independently trained on progressively higher-fidelity source or anchor domains and then evaluated on the most complex anchor domain and the def2-TZVP target domain.
Target-domain labels are not used for model training or model selection.
In the sliding-resistance example, the models is trained on the empirical-force-field source domain and directly deployed on MLFFs anchor-domain and DFT target-domain samples.
The PT indicators are calculated from source, anchor, and target latent representations before target-label-based performance assessment.
This protocol separates transfer risk diagnosis from final target-domain error measurement and provides a direct test of whether a representation learned from low-fidelity data can extrapolate to higher-fidelity physical domains.

\bibliography{main_text}

\end{document}